\documentclass[12pt,hyphens]{amsart}

\pdfoutput=1

\usepackage[utf8]{inputenc}
\usepackage{amssymb}
\usepackage{color}
\usepackage{mathrsfs}
\usepackage{comment}
\usepackage{hyperref}
\usepackage{tikz}
\usepackage{caption}
\usepackage{subcaption}

\newcommand{\abs}[1]{\left\lvert #1 \right\rvert}

\newcommand{\ang}[1]{\left\langle #1 \right\rangle}

\newcommand{\eff}{\mathbb{F}}

\newcommand{\rnum}{\mathbb{R}}
\newcommand{\znum}{\mathbb{Z}}

\newcommand{\numberthis}{\addtocounter{equation}{1}\tag{\theequation}}
\newtheorem{theorem}[subsection]{Theorem}
\newtheorem{question}[subsection]{Question}
\newtheorem{lemma}[subsection]{Lemma}
\newtheorem{proposition}[subsection]{Proposition}

\newtheorem{corollary}[subsection]{Corollary}
\theoremstyle{definition}
\newtheorem{definition}[subsection]{Definition}
\newcommand{\ar}{\mathrm{AR}}
\newcommand{\gr}{\mathrm{GR}}
\newcommand{\sr}{\mathrm{SR}}
\newcommand{\pr}{\mathrm{PR}}
\newcommand{\tr}{\mathrm{TR}}
\newcommand{\rank}{\mathrm{rank}}
\DeclareMathOperator{\F}{\mathbb{F}}
\DeclareMathOperator{\E}{\mathbb{E}}
\DeclareMathOperator{\matn}{\langle n \rangle}
\newcommand{\cala}{\mathcal{A}}
\newcommand{\calm}{\mathcal{M}}
\newcommand{\K}{\mathbb{K}}
\newcommand{\bias}{\mathrm{bias}}
\newcommand{\codim}{\mathrm{codim}}

\begin{document}

\title{Notions of Tensor Rank}
\author{Mandar Juvekar}
\address{Boston University, Boston, MA}
\email{mandarj@bu.edu}

\author{Arian Nadjimzadah}
\address{UCLA, Los Angeles, CA}
\email{anad@math.ucla.edu}

\date{September 2022; Revised June 2023}

\begin{abstract}
    Tensors, or multi-linear forms, are important objects in a variety of areas from analytics, to combinatorics, to computational complexity theory.
    Notions of tensor rank aim to quantify the ``complexity'' of these forms, and are thus also important.
    While there is one single definition of rank that completely captures the complexity of matrices (and thus linear transformations), there is no definitive analog for tensors.
    Rather, many notions of tensor rank have been defined over the years, each with their own set of uses.
    In this paper we survey the popular notions of tensor rank.
    We give a brief history of their introduction, motivating their existence, and discuss some of their applications in computer science.
    We also give proof sketches of recent results by Lovett, and Cohen and Moshkovitz, which prove asymptotic equivalence between three key notions of tensor rank over finite fields with at least three elements.
\end{abstract}

\maketitle

\section{Introduction} \label{sec:intro}
We first come across the notion of rank in a course on linear algebra. If $A$ is an $m \times n$ matrix over a field $\F$, the rank of $A$ is the dimension of the space spanned by its rows (or columns). To help generalize this definition, it will be useful to reinterpret the matrix $A$ as a bilinear form $T: \F^m \times \F^n \to \F$ by the natural correspondence
\begin{equation*}
    (x,y) \in \F^m \times \F^n \mapsto T(x,y) = \sum_{i,j} A_{i,j} x^i y^j.
\end{equation*}
As a notational convention, here and elsewhere in this paper we will use superscripts to index coordinates of a vector, and subscripts to index over different vectors. For example, if $x_1, x_2, \ldots$ are vectors, $x_5^7$ would be the seventh coordinate of the fifth vector.
Getting back to ranks, using this correspondence, we can formulate an alternative definition of rank where the rank of $T$ is the minimal natural number $r$ such that $T$ can be written as a sum of $r$ bilinear forms of ``lowest'' complexity, or \textit{rank 1 matrices}.
The natural objects of lowest complexity are the linear 1-forms, i.e., the dot product with a fixed vector.
Products of 1-forms $T_1(x)T_2(y)$ are then bilinear forms.
Since every bilinear form can be written as a finite sum of forms of this type, we set them to be our rank 1 bilinear forms (matrices).
Putting everything together, our alternative definition then says that the rank of a matrix $A$ is the minimum natural number $r$ such that the bilinear form $T$ corresponding to $A$ can be written as the sum of $r$ bilinear forms each of the type $T_1(x)T_2(y)$ where $T_1$ and $T_2$ are 1-forms.
This alternative definition agrees with the usual linear algebra definition of rank (see for instance \cite[Sections 32 and 51]{halmos:fin-dim}).

In this paper we focus on generalizing the notion of rank to higher-dimensional analogs of matrices. What are these higher-dimensional matrices?
The analogs in higher dimensions that we will look at are the \textit{$d$-tensors}.
\begin{definition}[$d$-tensor]
    Let $\F$ be a field. A \emph{$d$-tensor} $T$ on $\F^{n_1} \times \cdots \times \F^{n_d}$ is a multilinear map $T : \F^{n_1} \times \cdots \times \F^{n_d} \to \F$.
\end{definition}
Note that a $d$-tensor $T : \F^{n_1} \times \cdots \times \F^{n_d} \to \F$ can naturally be identified with a $d$-dimensional array $M$ such that
\begin{equation*}
    T(x_1, \ldots x_d) = \sum_{i_1 \in [n_1], \ldots, i_d \in [n_d]} M_{i_1, \ldots, i_d} x_1^{i_1} \cdots x_d^{i_d}.
\end{equation*}
So, following the convention for matrices, we will often denote the space of $d$-tensors on $\F^{n_1} \times \cdots \times \F^{n_d}$ by $\F^{n_1 \times \cdots \times n_d}$.\footnote{Here and elsewhere in the paper we will use $[n]$ to denote the set $\{1, 2, \ldots, n\}$.}

The definition of rank given above generalizes well to arbitrary $d$-tensors.
Setting the rank 1 $d$-tensors to products of $d$ linear 1-forms (or, since linear 1-forms are the same as 1-tensors, $d$ 1-tensors) leads to the notion of rank that is traditionally associated with tensors. We call it \emph{traditional rank}, or $\tr$ for short.
\begin{definition}[Traditional Rank] \label{def:tr}
    A $d$-tensor $T$ has \emph{traditional rank 1} if we can write
    \begin{equation*}
        T(x_1, \ldots, x_d) = T_1(x_1)T_2(x_2)\cdots T_d(x_d),
    \end{equation*}
    where each $T_j$ is a 1-tensor.

    The \emph{traditional rank} of an arbitrary $d$-tensor $T$, $\tr(T)$, is the minimum number $r$ such that we can write
    \begin{equation*}
        T(x) = \sum_{i=1}^r T_i(x)
    \end{equation*}
    where each $T_i$ is a $d$-tensor with traditional rank 1.
\end{definition}

While this definition does generalize matrix rank it turns out that there are also other, non-equivalent generalizations of matrix rank to arbitrary tensors that characterize the combinatorial, analytic, and geometric properties of tensors.
Each of these notions is useful in its own way.
Finding tight relationships between these notions remains an open research question.

One disadvantage of the traditional notion of tensor rank is that it is prohibitively hard to compute in general.
Since H{\aa}stad's work in 1989~\cite{has:tens-np} it has been known that computing traditional tensor rank over finite fields is NP-complete, and over $\mathbb Q$ is NP-hard.
In 2013, Hillar and Lim~\cite{hil-lim:most-tensor} showed that H{\aa}stad's proof could be modified to show that traditional rank is NP-hard over $\rnum$ and $\mathbb C$ as well.
Traditional rank also turns out to be NP-hard to approximate with arbitrarily small error bounds~\cite{swer:tens-hard}.
The non-traditional notions we will discuss are less ``stringent'' than traditional rank, so it is possible that they are easier to compute.
Showing (exact or approximate) hardness results for those notions, however, is still an open problem.

In this paper we introduce the landscape of notions of tensor rank, motivating their existence (Section~\ref{sec:ages}).
We then give examples of applications of these notions in computational complexity theory (Section~\ref{sec:apps}).
In Section~\ref{sec:triv} we give a rundown of the trivial relationships between the notions that come from their definitions.
Sections~\ref{sec:lov} and \ref{sec:mosh} give proof sketches for results by Lovett~\cite{lov:arank-applications} and Cohen and Moshkovitz~\cite{coh-mos:struct-bilin} respectively which, put together, prove asymptotic equivalence between three key notions of rank over finite fields with three or more elements.

\section{Tensor Rank Through the Ages} \label{sec:ages}

The oldest non-traditional notion of tensor rank is the so-called \textit{analytic rank}.
\begin{definition}[Analytic Rank] \label{def:ar}
    Let $\F$ be a finite field.
    The \emph{bias} of a $d$-tensor $T \in \F^{n_1 \times \cdots \times n_d}$ is given by
    \begin{equation*}
        \bias(T) = \E_{(x_1,\ldots, x_{d}) \in \F^{n_1}\times \cdots \times \F^{n_d}} \chi(T(x_1,\ldots, x_d)),
    \end{equation*}
    where $\chi$ is a nontrivial additive character which, in the case that $\F = \F_p$, can be taken to be
    $\chi(x) = e^{2\pi i x/p}$.

    The \emph{analytic rank} of $T$, denoted by $\ar(T)$, is then given by
    \begin{equation*}
        \ar(T) = -\log_{\abs{\F}} \bias(T).
    \end{equation*}

\end{definition}
This measure of the complexity of a tensor was first introduced by Gowers and Wolf in the context of higher-order Fourier analysis~\cite{gowers-wolf}. In Section \ref{sec:lov} we will see the connection between the analytic rank and the other combinatorial and geometric notions described below.

Going back to the traditional definition of tensor rank, there is no reason why one cannot use other kinds of tensors as rank 1 tensors, as long as the new notion agrees with the old in the case of matrices (in order to be a true generalization).
Doing so, it turns out, is useful for a whole host of applications.
The first such ``alternative'' notion was introduced in 2016 in the context of the so-called ``capset problem.''

A \textit{capset} in $\eff_3$ is a subset of $\eff_3$ with no non-trivial three-term arithmetic progressions.
That is, a capset is a set $A \subseteq \eff_3$ that does not contain $\{x, x + r, x + 2r\}$ for any $x, r \in \eff_3$ with $r \neq 0$.
The capset problem asks what the maximum size of a capset in $\eff_3$ can be. In the spring of 2016, Croot, Lev, and Pach~\cite{cro-lev-pac:prog-free} proved a breakthrough result for a similar problem in the additive group $\znum/4\znum$.
They showed that if $A \subseteq (\znum/4\znum)^n$ contains no non-trivial three-term arithmetic progressions, then $\abs{A} \leq 3.60172^n$.
Soon after, Ellenberg and Gijswijt~\cite{ell-gij:large-subsets} generalized the Croot-Lev-Pach argument to show that any progression-free set $A \subseteq (\znum/p\znum)^n$, where $p$ is a prime, satisfies $\abs{A} \leq (J(p)p)^n$ where $J(p)$ is an explicit constant less than 1.
This exponentially improved the known bound for the capset problem, bringing it down from $O(3^n/n^{1+c})$ (where $c$ is some absolute constant) to $O(2.756^n)$.
Tao~\cite{tao:symm-form} reformulated the Ellenberg-Gijswijt argument in terms of tensors, introducing what is now known as the \textit{slice rank} of a tensor.

\begin{definition}[Slice Rank] \label{def:sr}
    A $d$-tensor $T$ has \emph{slice rank 1} if we can write
    \begin{equation*}
        T(x_1, \ldots, x_d) = T_1(x_i) T_2(x_j : j \neq i),
    \end{equation*}
    where $i \in [d]$, $T_1$ is a 1-tensor, and $T_2$ is a $(d-1)$-tensor.

    The \emph{slice rank} of an arbitrary $d$-tensor $T$, $\sr(T)$, is the minimum number $r$ such that we can write
    \begin{equation*}
        T(x) = \sum_{i=1}^r T_i(x)
    \end{equation*}
    where each $T_i$ is a $d$-tensor with slice rank 1.
\end{definition}

Notice that in essence this definition just modifies Definition~\ref{def:tr} so that our rank 1 tensors go from being products of $d$ 1-tensors to being the product of two tensors: one of order 1 and one of order $d-1$.
Also notice that in the case of matrices (which are 2-tensors), Definitions~\ref{def:sr} and \ref{def:tr} are exactly the same and agree with the usual definition of matrix rank.

Slice rank, as we defined above, defines its rank 1 tensors by ``slicing'' off one coordinate and multiplying a 1-tensor applied to that coordinate with a $(d-1)$-tensor applied to the remaining coordinates.
A natural extension to this is to, instead of slicing off one coordinate, allow arbitrary partitions of the coordinates into two parts.
This was done in 2017 (preprint; published in 2020) by Naslund~\cite{nas:part-rank}, who called this new notion of rank \textit{partition rank}.
We formally define partition rank as follows.

\begin{definition}[Partition Rank] \label{def:pr}
    A $d$-tensor $T$ has \emph{partition rank 1} if we can write
    \begin{equation*}
        T(x_1, \ldots, x_d) = T_1(x_i : i \in S) T_2(x_j : j \notin S),
    \end{equation*}
    where $S \subset [d]$ with $1 \leq \abs{S} < d$, $T_1$ is a $\abs{S}$-tensor, and $T_2$ is a $(d - \abs{S})$-tensor.

    The \emph{partition rank} of an arbitrary $d$-tensor $T$, $\pr(T)$, is the minimum number $r$ such that we can write
    \begin{equation*}
        T(x) = \sum_{i=1}^r T_i(x)
    \end{equation*}
    where each $T_i$ is a $d$-tensor with partition rank 1.
\end{definition}

Using this new notion of rank, Naslund showed that any set $A \subseteq \eff_q^n$ with size at least ${n + (k-1)q \choose (k-1)(q-1)}$ must have distinct vectors $x_1, x_2, \ldots, x_{k+1}$ such that the vectors $x_1 - x_{k+1}, x_2 - x_{k+1}, \ldots, x_k - x_{k+1}$ are mutually orthogonal (he called a collection of such $x_1, \ldots, x_{k+1}$ a \textit{$k$-right corner}).

Our last, most recent notion of rank is motivated by an age-old, fundamental question in computer science: asymptotically how many scalar additions and multiplications are necessary to multiply two $n \times n$ matrices?
Tensors are naturally connected to this problem because the operation of matrix multiplication itself can be thought of as a tensor.

Recall that matrices $M$ (2-tensors) can simultaneously be thought of as bilinear forms $T: \F^m \times \F^n \to \F$ that map $(x, y) \mapsto x^T M y$ and linear maps $T': \F^m \to \F^n$ that map $x \mapsto Mx$.
Similarly, we can also think of $d$-tensors
\begin{equation*}
T(x_1, \ldots x_d) = \sum_{i_1 \in [n_1], \ldots, i_d \in [n_d]} T_{i_1, \ldots, i_d} x_1^{i_1} \cdots x_d^{i_d}
\end{equation*}
as $(d-1)$-linear maps $T': \F^{n_1} \times \cdots \times \F^{n_{d-1}} \to \F^{n_d}$ given by
\begin{equation*}
    [T'(x_1,\ldots, x_{d-1})]_k =
    \sum_{i_1 \in [n_1],\ldots, i_{d-1}\in [n_{d-1}]} T_{i_1, \ldots, i_{d-1}, k} x_1^{i_1} \cdots x_{d-1}^{i_{d-1}}.
\end{equation*}

By the usual formula for matrix multiplication, if $M$ and $N$ are $n \times n$ matrices over $\F$, then $MN$ is an $n \times n$ matrix with
\begin{equation*}
    (MN)_{ij} = \sum_{\ell = 1}^n M_{i\ell} N_{\ell j}.
\end{equation*}
Identifying the space of $n \times n$ matrices with $\F^{n^2}$, the operation of matrix multiplication takes two elements of $\F^{n^2}$ to another element of $\F^{n^2}$ where each entry is a bilinear form applied to $M$ and $N$.
Thus matrix multiplication is a 3-tensor over the space of $n \times n$ matrices.
As we discuss in Section~\ref{sec:apps} the rank of this tensor is intimately related to the computational complexity of matrix multiplication.

Motivated by this application (as well as a few others), Kopparty, Moshkovitz, and Zuiddam~\cite{kop-mos-zui:gr-subrank} recently introduced the notion of \textit{geometric rank}.
Unlike the other types of ranks discussed, this notion of rank does not aim to capture the combinatorial or analytic properties of the tensor.
Rather, it looks at the geometric properties of the tensor, defining rank as the codimension of an algebraic variety.
Kopparty, Moshkovitz, and Zuiddam used this new notion of rank to prove tight bounds about the subrank of the matrix multiplication tensor (the subrank is a quantity related to the rank that is useful in the computational complexity analysis for matrix multiplication).

The geometric rank of a tensor is formally defined as follows.
\begin{definition}[Geometric Rank] \label{def:gr}
    Let $T \in \F^{n_1 \times \cdots \times n_d}$ be a $d$-tensor with $d \geq 2$. The \emph{geometric rank} of $T$ is
    \begin{align*}
        \gr(T) = \codim\{(x_1, \ldots, x_{d-1}) &\in \F^{n_1} \times \cdots \times \F^{n_{d-1}} \\
        &: \forall z \in \F^{n_d}, \ T(x_1, \ldots, x_{d-1}, z) = 0\}.
    \end{align*}
\end{definition}

Here we use the usual definition of the codimension of an algebraic variety.
If $V \subseteq \F^n$ is an algebraic variety (that is possibly reducible), the dimension of $V$, $\dim(V)$, is defined to be the length of a maximal chain of irreducible subvarieties of $V$. The codimension, $\codim(V)$, is then defined to be $n-~\dim(V)$. For more detailed explanations of these concepts we recommend looking at~\cite{harris:alg-geom}.

Geometric rank---as Kopparty, Moshkovitz, and Zuiddam mention without proof---coincides with the linear algebra definition of matrix rank when $d = 2$.
We give a quick proof of that here.

\begin{proposition}
    Let $M \in \F^{m\times n}$ be a matrix (a 2-tensor), and let $r$ be the rank of $M$ as per the linear algebra definition.
    Then $r = \gr(M)$.
\end{proposition}

\begin{proof}
    By definition of geometric rank,
    \begin{align*}
        \gr(M) &= \codim\{x \in \F^m : \forall y \in \F^n x^T M y = 0\} \\
        &= \codim\{x \in \F^m : x^T M = \mathbf{0}^T\} \\
        &= m - \dim\{x \in \F^m : x^T M = \mathbf{0}^T\},
    \end{align*}
    where $\bf{0}$ is the zero vector in $\F^m$.
    Notice that $\dim\{x \in \F^m : x^T M \equiv \mathbf{0}^T\}$ is just the nullity of $M$, and so by the rank nullity theorem we are done.
\end{proof}

There is a slight caveat to Definition~\ref{def:gr}, which is that it implicitly assumes $\F$ to be algebraically closed.
This, however, is not a problem because in a general field we can extend this definition via the embedding of the field into its algebraic closure.

\section{Applications in Computer Science} \label{sec:apps}

Uttering the words ``tensor rank'' to a computer scientist will, with high probability, elicit a response that contains the phrase ``matrix multiplication.''
And for good reason; notions of tensor rank happen to be intimately related to the complexity and efficient computation of matrix multiplication.

Multiplying two matrices is a fundamental computation with applications in nearly every field computers are used in.
It is thus not surprising that the following question is of great interest.
\begin{question}[Arithmetic Complexity of Matrix Multiplication] \label{q:ac-mm}
    Given two $n \times n$ matrices $A$ and $B$ over a field $\F$, what is the minimum number of addition and scalar multiplication operations needed to compute the $n \times n$ product matrix $AB$ where
    \begin{equation} \label{eq:mm}
        (AB)_{ij} = \sum_{\ell = 1}^n A_{i\ell} B_{\ell j}?
    \end{equation}
    Can one find an algorithm that actually computes the product with that number of operations?
\end{question}

The number of operations in Question~\ref{q:ac-mm} is known as the \textit{arithmetic complexity} of $n \times n$ matrix multiplication.
Let us call this number $\cala(n)$.
An exact determination of $\cala(n)$ seems to be outside the range of methods available at the present time, so much of the work around this has been focused on getting asymptotic bounds on arithmetic complexity.
It is then useful to define the following quantity.

\begin{definition}[Exponent of Matrix Multiplication] \label{def:exponent}
    The \emph{exponent of matrix multiplication}, $\omega$, is defined as
    \begin{equation*}
        \omega = \inf\{\tau \in \rnum : \cala(n) = O(n^\tau)\}.
    \end{equation*}
    Equivalently,
    \begin{align*}
         \omega = \inf\{\tau : \exists &\text{ an algorithm to compute } \\
         &n \times n \text{ MM with } O(n^\tau) \text{ operations}\}.
    \end{align*}
\end{definition}

It is easy to see that $\omega \in [2,3]$.
The upper bound comes from the fact that the na\"ive algorithm of computing each entry of the product using Equation~\ref{eq:mm} is $O(n^3)$.
The lower bound holds because any algorithm for matrix multiplication would have to make at least $n^2$ computations since the output is an $n \times n$ matrix.
A long-standing conjecture in algebraic complexity theory is that $\omega$ is in fact equal to 2.

As mentioned in Section~\ref{sec:ages}, the operation of matrix multiplication is itself a 3-tensor.
It is convention to denote the tensor to multiply a $k \times m$ matrix and an $m \times n$ matrix with $\ang{k,m,n}$.
Since we are only focusing on square matrices, we will focus on the tensors $\ang{n,n,n}$, which we will often shorten to $\matn$.
Seeing as rank is supposed to encode the complexity of a tensor, it is not surprising that it is closely connected to $\cala(n)$.
The following result is known to hold.

\begin{theorem} Over any field, we have \label{t:mm-rank}
    \begin{equation*}
        \omega = \inf\{\tau \in \rnum : \tr(\matn) = O(n^\tau)\}.
    \end{equation*}
\end{theorem}
\begin{proof}
    See \cite[Proposition 15.1]{bur-cla-sho:alg-compl}
\end{proof}

An important application of Theorem~\ref{t:mm-rank} is in showing that the exponent of matrix multiplication does not change under another important, related measure of complexity called the \textit{multiplicative complexity} of matrix multiplication.
The multiplicative complexity of matrix multiplication is defined as the minimum number of scalar multiplication operations required for $n \times n$ matrix multiplication (the algorithm is of course also allowed additions).
Let us call this quantity $\calm(n)$.
It has been shown that $\calm$ relates to traditional tensor rank in the following way (see, for instance, \cite[Section 14.1]{bur-cla-sho:alg-compl}).
\begin{equation*}
    \frac{1}{2}\tr(\matn) \leq \calm(n) \leq \tr(\matn).
\end{equation*}
From this inequality it follows that $\tr(\matn) = O(n^\tau)$ implies $\calm(n) = O(n^\tau)$ and $\calm(n) = O(n^\tau)$ implies $\tr(\matn) = O(n^\tau)$.
Combining this with Theorem~\ref{t:mm-rank} gives us $\omega = \inf\{\tau \in \rnum : \calm(n) = O(n^\tau)\}$.

The connection in Theorem~\ref{t:mm-rank} also implies that upper and lower bounds on $\tr$ directly translate to upper and lower bounds on the exponent of matrix multiplication.
The following result makes this explicit.

\begin{theorem}[\cite{bla:fast-mm}] \label{t:mm-tr-ub}
    If $\tr(\ang{k,m,n}) \leq r$, then $\omega \leq \log_{kmn} r$.
\end{theorem}

In his seminal work, Strassen~\cite{str:gauss} showed that square matrices of size 2 can be multiplied with seven multiplication operations (as opposed to eight using the na\"ive method) at the cost of a few more addition operations.
It follows from his work that $\tr(\ang{2}) \leq 7$.
Plugging this into Theorem~\ref{t:mm-tr-ub} tells us that $\omega \leq 3 \log_{8} 7 = \log_2 7 = 2.801... \leq 2.81$.
Hopcroft and Kerr~\cite{hop-ker:min-num-mult} and Winograd~\cite{win:on-mult} later (independently) showed that $2 \times 2$ multiplication is not possible with just six multiplications, which implies that $\tr(\ang{2})$ is in fact equal to 7.
Thus $\log_2 7$ is the sharpest bound one can get from Theorem~\ref{t:mm-tr-ub} using $\ang{2}$.
It turns out that $\tr(\ang{70}) \leq 143.640$~\cite{pan:fast-matrix}, which gives a slightly better bound of 2.80.
Whether other tensors can lead to sharper bounds remains an open question.

Strassen~\cite{str:asymp-spectr}, using more traditional-rank-related techniques, was able to prove that $\omega < 2.48$.
Coppersmith and Winograd~\cite{cop-win:matrix-mult}, using a construction of arithmetic-progression-free sets, showed that $\omega \leq 2.375...$.
Starting in 2010, by analyzing higher-order variants of the Coppersmith-Winograd construction, Stothers~\cite{stothers:diss}, then Vassilevska Willians~\cite{wil:mult-matr}, and then Le Gall~\cite{legall:powers-tensors} made incremental improvements.
This led to the current best upper bound of $\omega < 2.372...$.

The connections between matrix multiplication and tensor rank do not end here.
As mentioned in the previous paragraph, Coppersmith and Winograd's proof involved a construction of arithmetic-progression-free sets.
This suggests that the complexity of matrix multiplication might be related to the capset problem and Tao's slice rank.
And indeed, that did turn out to be the case.
In 2003, Cohn and Umans~\cite{CohnUmans} described a framework for proving upper bounds on $\omega$ that involves reducing matrix multiplication to group algebra multiplication.
In 2012, Alon, Shpilka, and Umans~\cite{alon2012sunflowers} proved relations between various then-open conjectures in combinatorics and bounds on $\omega$.
The resolution of the capset problem in 2016 settled some of the conjectures involved in the 2012 work.
That led to the 2017 work of Blasiak et al.~\cite{Blasiak_2017}, which used the resolution of the capset problem to rule out obtaining an $\omega = 2$ using a subclass of Cohn-Umans-style constructions.
In doing so, they extended the capset result, making extensive use of the notion of slice ranks.

Finally, the introduction of the geometric rank was in part motivated by the complexity of matrix multiplication.
Kopparty, Moshkovitz, and Zuiddam~\cite{kop-mos-zui:gr-subrank} used the notion of geometric rank to prove a tight upper bound on the so-called \textit{border subrank} of the matrix multiplication tensor, which matched a known lower bound.
While the exact definitions of subrank are beyond the scope of this survey, we recommend reading \cite{kop-mos-zui:gr-subrank} for further details.

\section{Some Trivial Relationships} \label{sec:triv}
Before getting to the deeper connections between the different notions of rank, we state some consequences obvious from the definitions. Any rank 1 form used in Definition \ref{def:sr},
\begin{equation*}
    T_1(x_1)T_2(x_j : j\neq i),
\end{equation*}
is in particular a rank 1 form of the form in Definition \ref{def:pr},
\begin{equation}\label{eq:trivial}
    T_1(x_i : i\in S)T_2(x_j : j\notin S).
\end{equation}
A tensor of the form in \eqref{eq:trivial} is in turn of the form of the rank 1 form in Definition \ref{def:tr},
\begin{equation*}
    T_1(x_1)\cdots T_d(x_d).
\end{equation*}
This implies the following result.
\begin{proposition} \label{prop:trivial}
    For any tensor $T$,
    \begin{equation*}
        \pr(T) \leq \sr(T) \leq \tr(T).
    \end{equation*}
\end{proposition}

\section{Relating Analytic Rank to the Others} \label{sec:lov}

We now discuss some deeper relationships which will show a cycle of relationships between $\ar$, $\sr$, and $\tr$.
We start with an elegant result due to Lovett relating the analytic and partition rank.
For ease of notation we will look at tensors $T: V^d \to \F$ where $V = \F^n$ for some $n$.
It will be easy to see that this proof can be generalized to arbitrary $T \in \F^{n_1 \times \cdots \times n_d}$.

\begin{theorem}[\cite{lov:arank-applications}] \label{t:ar-pr}
    Let $T: V^d \to \F$ be a $d$-tensor.
    Then $\ar(T) \leq \pr(T)$.
\end{theorem}

To prove this theorem we will need a lemma as well as another theorem about the arithmetic rank. Both of these are due to Lovett.

\begin{lemma}[\cite{lov:arank-applications}] \label{lem:lov}
    For each $I \subseteq [d]$ let $R_I : V^{\abs{I}} \to \F$ be an $\abs{I}$-tensor. Let
    \begin{equation*}
        R(x) = \sum_{I \subseteq [d]} R_I(x_j : j \in I).
    \end{equation*}
    Then $\abs{\bias(R)} \leq \bias(R_{[d]})$.
\end{lemma}

\begin{proof}
    First, let $W_0, \ldots, W_n$ be arbitrary functions from $\F^m$ to $\F$. Let $A$ and $B$ be functions from $\F^n \times \F^m$ to $\F$ defined by
    \begin{align*}
        A(x,y) &= \sum_{i=1}^n x^i W_i(y), \text{ and} \\
        B(x,y) &= A(x,y) + W_0(y).
    \end{align*}
    From the definition of bias we get
    \begin{align*}
        \bias(B) &= \E_{x,y} \chi(B(x,y)) \\
        &= \E_{x,y} \chi(A(x,y) + W_0(y)) \\
        &= \E_{x,y} [\chi(A(x,y)) \cdot \chi(W_0(y))] \\
        &= \E_y [\chi(W_0(y)) \cdot \E_x \chi(\sum_{i=1}^n x^i W_i(y))] \\
        &= \E_y [1_{W_1(y) = \cdots = W_n(y) = 0} \chi(W_0(y))],
    \end{align*}
    where the last equality holds because the expectation over $x$ is 0 whenever any of the $W_j$ with $1 \leq j \leq n$ are non-zero.
    Taking absolute values and using the triangle inequality gives us
    \begin{equation} \label{eq:lemma}
        \abs{\bias(B)} \leq \E_y [1_{W_1(y) = \cdots = W_n(y) = 0}] = \bias(A).
    \end{equation}

    With this smaller result, we prove the lemma by applying it repeatedly.
    Fixing $i \in [d]$, we can break up the summation in the definition of $R$ as follows.
    \begin{equation*}
        R(x) = \sum_{I \subseteq [d] ; i \in I} R_I(x_j : j \in I) + \sum_{I \subseteq [d] ; i \notin I} R_I(x_j : j \in I).
    \end{equation*}
    Notice that each $R_I(x_j : j \in I)$ where $i \in I$ is a tensor depending on $x_i$, and thus can be written as $x_i^j W_j(x_k : k \neq i)$. Thus the first sum is of the form $\sum_j x_i^j W_j(y)$, which matches the form of $A$ above. The second sum does not depend on $x_i$, and so is of the form of $B$.
    Thus, applying Equation~\ref{eq:lemma} gives us
    \begin{equation*}
        \abs{\bias(R)} \leq \bias\left(\sum_{I \subseteq [d];i \in I} R_I(x_j : j \in I)\right).
    \end{equation*}
    We use this inequality iteratively.
    First, applying it for $i = d$ we have
    \begin{equation*}
        \abs{\bias(R)} \leq \bias\left(\sum_{I \subseteq [d];d \in I} R_I(x_j : j \in I)\right).
    \end{equation*}
    Using $\sum_{I \subseteq [d];d \in I} R_I(x_j : x_j \in I)$ instead of the tensor $R$ in our inequality with $i = d-1$ gives us
    \begin{equation*}
        \abs{\bias\left(\sum_{I \subseteq [d];d \in I} R_I(x_j : j \in I)\right)} \leq \bias\left(\sum_{I \subseteq [d-1];d-1 \in I; d \in I} R_I(x_k : k \in I)\right).
    \end{equation*}
    Continuing this process and applying the inequality for $d-2, \ldots, 1$, and then chaining the inequalities gives us
    \begin{align*}
        \abs{\bias(R)} \leq \bias(R_{[d]})
    \end{align*}
    as desired.
\end{proof}

Using this lemma, Lovett proves that the analytic rank is sub-additive.
\begin{theorem}[\cite{lov:arank-applications}] \label{t:ar-subadd}
    Let $T, S : V^d \to \F$ be $d$-tensors. Then
    \begin{equation*}
        \ar(T + S) \leq \ar(T) + \ar(S).
    \end{equation*}
\end{theorem}

\begin{proof}[Proof Sketch]
    For this proof, Lovett defines functions $T_I$ and $S_I$ such that for any $x, y$, $T(x, y) = \sum_{I \subseteq [d]} T_I(x_I, y_{[d] - I})$ (similarly for $S$) where $x_I = (x_i : i \in I)$ (similarly for $y$).
    Using this decomposition and performing some algebra leads to
    \begin{align*}
        \bias(T)\bias(S) \leq \abs{\bias\left((T+S)(x) + \sum_{I \subsetneq [d]} S_I(x_I, b_{[d] - I})\right)},
    \end{align*}
    where $b$ is a fixed choice for $y$.

    Applying Lemma~\ref{lem:lov} to the functions $R_{[d]} = (T + S)(x)$ and $R_I = S_I(x_I, b_{[d] - I})$ shows that the right hand side of the inequality is less than or equal to $\bias(R{[d]})$, which in our case equals $\bias(T+S)$.
    Putting everything together we get
    \begin{equation*}
        \bias(T+S) \geq \bias(T)\bias(S),
    \end{equation*}
    and so by the definition of analytic rank,
    \begin{align*}
        \ar(T+S) \leq \ar(T) + \ar(S).
    \end{align*}
\end{proof}

Using these two results, we can now prove Theorem~\ref{t:ar-pr}.

\begin{proof}[Proof of Theorem~\ref{t:ar-pr}]
    Given Theorem~\ref{t:ar-subadd}, it suffices to prove Theorem~\ref{t:ar-pr} for tensors of partition rank 1.
    So suppose $T : V^d \to \F$ has partition rank 1.
    Then we can find a partition $A \sqcup B = [d]$ with $\abs{A}, \abs{B} \geq 1$ such that
    \begin{equation*}
        T(x) = T_1(x_i : i \in A) T_2(x_j : j \in B).
    \end{equation*}
    For convenience, let us denote $(x_i : i \in A)$ by $x_A$ and $(x_j : j \in B)$ by $x_B$.

    Since the partition rank of $T$ is 1, we need to show that $\ar(T) \leq 1$. To do so, it suffices to show (by definition of analytic rank) that $\bias(T) \geq \abs{\F}^{-1}$.
    For any $a, b \in \F$ define the function
    \begin{align*}
        F_{a,b}(x) = (T_1(x_A) + a)(T_2(x_B) + b).
    \end{align*}
    Expanding it out we can write
    \begin{align*}
        F_{a,b}(x) &= T_1(x_A)T_2(x_B) + T_1(x_A)b + T_2(x_B)a + ab \\
        &= T(x) + T_1(x_A)b + T_2(x_B)a + ab.
    \end{align*}
    Letting $R_A(x) = T_1(x_A)b$, $R_B(x) = T_2(x_B)a$, $R_\emptyset = ab$, and $R_{[d]} = T$, we can apply Lemma~\ref{lem:lov} to get
    \begin{align*}
        \abs{\bias(F_{a,b})} \leq \bias(T).
    \end{align*}
    On the other hand, if $a$ and $b$ are chosen uniformly, we can take the expectation of the bias over $a$ and $b$ to get
    \begin{align*}
        \E_{a,b} \bias(F_{a,b}) &= \E_{a,b \in F; x \in V^d}[\chi((T_1(x_A) + a)(T_2(x_B) + b))] \\
        &= \E_{a,b \in \F} [\chi(ab)] \\
        &= \mathrm{Pr}_{b \in \F} [b = 0] \\
        &= \abs{\F}^{-1}.
    \end{align*}
    This proves our theorem.
\end{proof}

Combining Theorem~\ref{t:ar-pr} with Proposition~\ref{prop:trivial} gives us the following corollary.

\begin{corollary}
    Let $T : V^d \to \F$ be a $d$-tensor. Then $\ar(T) \leq \sr(T)$.
\end{corollary}

\section{Closing the Loop for 3-tensors: Cohen and Moshkovitz's Argument} \label{sec:mosh}

Cohen and Moshkovitz~\cite{coh-mos:struct-bilin} proved the following two theorems which, for 3-tensors over finite fields with at least 3 elements, show asymptotic equivalence between $\ar$, $\gr$, and $\sr$ when combined with Lovett's result.

\begin{theorem}[\cite{coh-mos:struct-bilin}]\label{thm:sr-gr}
    For any $3$-tensor $T$ over a perfect field $\F$,
    \begin{equation*}
        \sr(T) \leq 3\gr(T)
    \end{equation*}
\end{theorem}

\begin{theorem}[\cite{coh-mos:struct-bilin}]\label{thm:ar-gr}
    For any 3-tensor $T$ over any finite field $\F$,
    \begin{equation*}
        \ar(T) \geq (1 - \log_{\abs{\F}} 2)\gr(T).
    \end{equation*}
\end{theorem}

Our goal will be to give an intuitive sketch of their argument.
Techniques from algebraic geometry will be essential, so we give some background.

\subsection{Background from Algebraic Geometry}
From now on, let $\K$ be a field.
Recall that a \emph{variety cut out by a (finite) number of polynomials} over $\K^n$ is the subspace of $\K^n$ where all the polynomials vanish.
The ideal $I(V)$ of a variety $V$ is defined as follows.
\begin{definition}[Ideal of a Variety]
Let $V \subset \K^n$ be a variety. The \emph{ideal} of $V$ is
\begin{equation*}
     I(V) = \{f \in \K[x] : f(p) = 0 \text{ for each } p \in V \}.
\end{equation*}

\end{definition}
Now we define the tangent space $T_p V$ to a variety $V$ at $p$.
\begin{definition}[Tangent Space]
    Let $V \subset \K^n$ be a variety. The \emph{tangent space} $T_p V$ to a variety $V$ at $p$ is
    \begin{equation*}
        T_p V = \left\{\mathbf v\in \K^n :  \frac{\partial g}{\partial \mathbf v}(p) = 0, \text{ for each } g \in I(V) \right\}.
    \end{equation*}
\end{definition}
One can check that any variety can be written uniquely as the union of \emph{irreducible varieties}. An irreducible variety is one which cannot be written as the union of strictly contained subvarieties. Then the \emph{dimension} of a variety $V$ can be defined as follows.
\begin{definition}[Dimension of a Variety]
    The \emph{dimension} of a variety $V \subset \K^n$, written $\dim V$, is the maximal length of a chain of irreducible varieties such that
    \begin{equation*}
        \emptyset \neq V_1 \subsetneq V_2 \subsetneq \cdots \subsetneq V_k \subsetneq V.
    \end{equation*}
    The \emph{codimension}
    of a variety $V \subset \K^n$, $\codim V$, is
    \begin{equation*}
        \codim V = n - \dim V.
    \end{equation*}
\end{definition}

\subsection{Proof Sketch of Theorem \ref{thm:sr-gr}}
For the full argument see \cite{coh-mos:struct-bilin}.
We will give a broader overview of the proof with the aim of conveying the intuition.

The proof of Theorem~\ref{thm:sr-gr} hinges on the following previously proved rephrasing of geometric rank.

\begin{lemma}\label{lem:gr-equiv-fact}
    For a 3-tensor over any field $\K$,
    \begin{equation*}
        \gr(T) = \min_r \{r + \codim\{x : \rank \ T(x, \cdot, \cdot) = r \} \}.
    \end{equation*}
\end{lemma}
\begin{proof}
    See \cite[Theorem 3.1]{kop-mos-zui:gr-subrank}.
\end{proof}

Let $r$ achieve the minimum in Lemma \ref{lem:gr-equiv-fact}.
Defining
\begin{equation*}
    X_r = \{x \in \overline\F^{n_1} :  \rank T(x,\cdot,\cdot) \leq r\},
\end{equation*}
one can use Lemma~\ref{lem:gr-equiv-fact} to show that $\gr(T) = r + \codim X_r$.
A 3-tensor $T \in \F^{n_1 \times n_2 \times n_3}$ with encoding array $(T_{i,j,k})$ can be reinterpreted as a linear space
$L \subset \overline \F^{n_1 \times n_2}$ spanned by $\{M_1,\ldots, M_{n_3}\}$, where
$M_k = (T_{i,j,k})_{i,j}$. One also has an association going the other way.
Given a matrix space $L$, choose a basis $M_1,\ldots, M_k$, and let $T$ be the tensor with encoding array $(M_1,\ldots, M_k)$. Thus the notions of rank for tensors can be transferred to matrix spaces.

Now let
\begin{equation*}
    L_r = L \cap M_r,
\end{equation*}
where $M_r$ is the space of $n_1 \times n_2$ matrices of rank at most $r$. Note that $M_r$, being the space cut out by the $(r+1) \times (r+1)$ minors, is itself a variety.
It turns out that the variety $X_r$ is isomorphic to $L_r \times \overline \F^{n_1 - d}$. Thus $\codim X_r = \codim_L L_r$.\footnote{Here $\codim_L X$ is the codimension of a variety $X$ in a linear space $L$, i.e. $\codim_L X = \dim L - \dim X$.}
We have reduced the problem to analyzing these ``slices'' of the matrix space $L$. Above, $\overline \F$ is the algebraic closure of $\F$, which will be easier to work over in the arguments to follow. We let $\overline \sr(T)$ be as in Definition \ref{def:sr}, except the coefficients are allowed to come from $\overline \F$.
Clearly $\overline \sr(T) \leq \sr(T)$. One also has $2\sr(T) \leq 3\overline{\sr}(T)$ which is proved in \cite{coh-mos:struct-bilin}.

We would be done if we had the following.
\begin{lemma}\label{lem:key-mosh}
    Let $L \subset \K^{m\times n}$ be a matrix space over any algebraically closed field $\K$. For any $r \in \mathbb{N}$,
    \begin{equation*}
        \sr(L) \leq 2r + \codim_L L_r.
    \end{equation*}
\end{lemma}
To see this, observe that
\begin{align*}
    \overline\sr(T) &= \sr(L) \leq 2r + \codim_L L_r \\
    &= 2r + \codim X_r \\
    &= 2\gr(T) - \codim X_r \\
    &\leq 2\gr(T).
\end{align*}
Applying $2\sr(T) \leq 3\overline \sr (T)$ gives $\sr(T) \leq 3\gr(T)$.

\begin{proof}[Proof sketch of Lemma \ref{lem:key-mosh}]
    We proceed by induction. The base case is $\sr(L) \leq \dim L$, which is not difficult to show. Now consider the inductive step.
    Let $V$ be an irreducible component of $L_r$ with $\dim V = \dim L_r$, and $A \in V \setminus M_{r-1}$.
    If $V \setminus M_{r-1} = \emptyset$ then we have $V \subseteq L_{r-1}$, and the result follows from induction.
    So we assume $V \setminus M_{r-1} \neq \emptyset$.
    The trick now is to decompose the matrix space $L$ into subspaces $P = L \cap T_A M_r$ and
    $P^\perp$. This particular decomposition is useful because one can run an explicit calculation to find that
    \begin{equation*}
    \sr(T_A M_r) \leq 2r.
    \end{equation*}
    We need the following simple properties of slice rank and tangent spaces to conclude the result from that calculation.
    \begin{lemma}\label{lem:simple-stuff} ~
    \begin{enumerate}
        \item $\sr(L) \leq \dim L$.
        \item $\sr(L') \leq \sr(L)$ if $L' \subset L$.
        \item $\sr(L + L') \leq \sr(L) + \sr(L')$.
        \item $T_p(V \cap W) \subset T_p V \cap T_p W$.
        \item If $V$ is irreducible and $p \in V$ then $\dim V \leq \dim T_p V$.
    \end{enumerate}
    \end{lemma}
    Using Lemma~\ref{lem:simple-stuff} and the main estimate $\sr(T_A M_r) \leq 2r$,
    \begin{equation*}
        \sr(P) = \sr(L \cap T_A M_r) \leq \sr(T_A M_r) \leq 2r.
    \end{equation*}
    We also have
    \begin{align*}
        \dim L_r &\leq \dim T_A V \\
        &\leq \dim T_A L_r \\
        &\leq \dim(T_A L \cap T_A M_r) \\
        &\leq \dim(L \cap T_A M_r) = \dim P.
    \end{align*}
    Thus we have $\codim_L P \leq \codim_L L_r$. Using Lemma~\ref{lem:simple-stuff} we get
    \begin{equation*}
        \sr(P^\perp) \leq \codim_L P \leq \codim_L L_r.
    \end{equation*}
    Hence
    \begin{align*}
        \sr(L) &= \sr(P + P^\perp) \\
        &\leq \sr(P) + \sr(P^\perp) \\
        &\leq 2r + \codim_L L_r.
    \end{align*}
\end{proof}

\subsection{Proof Sketch of Theorem~\ref{thm:ar-gr}}
    The key ingredient of the proof is the following generalization of the Schwarz-Zippel lemma. Define $V(\F) = V \cap \F^n$ for a variety $V \subset \overline\F^n$ defined over $\F$.
    \begin{lemma}\label{lem:schwarz-zippel}
    For any variety $V \subset \overline \F^n $ defined over any finite field $\F$ cut out by polynomials of degree at most $d$,
    \begin{equation*}
        \frac{|V(\F)|}{|\F|^n} \leq \left( \frac{d}{|\F|}\right)^{\codim V}.
    \end{equation*}
    \end{lemma}
    We use Lemma \ref{lem:schwarz-zippel} to prove Theorem \ref{thm:ar-gr}.
    Consider $T \in \F^{n_1\times n_2\times n_3}$, where we interpret $T$ as its defining array. Set $V = \ker(T) \subset \overline \F^N$ with $N = n_1 + n_2$, where here we use the bilinear map formulation of $T$. By Lemma \ref{lem:schwarz-zippel},
    \begin{equation*}
        \frac{|V(\F)|}{|\F|^n} \leq \left( \frac{d}{|\F|}\right)^{\codim V},
    \end{equation*}
    since $T$ is bilinear.
    Thus
    \begin{align*}
        \ar(T) &= -\log_{|\F|} \frac{|V(\F)|}{|\F|^N} \\
        &\geq \codim V \cdot (1-\log_{|\F|} 2) \tag{Using Lemma~\ref{lem:schwarz-zippel}} \\
        &= \gr(T) (1-\log_{|\F|} 2).
    \end{align*}
    The first equality above is true since
    \begin{align*}
        \bias(T) &= \E_{(x_1,x_2)\in\F^{n_1}\times \F{n_2}} \left[ \E_{x_3 \in F^{n_3}}
        \chi(T(x_1,x_2,x_3))\right] \\
        &=
        \Pr_{(x_1,x_2)\in \F^{n_1} \times \F{n_2}} [T(x_1,x_2, \cdot) \equiv 0 ] \numberthis \label{eq:prob}\\
        &= \frac{|V(\F)|}{|\F|^N}.
    \end{align*}
    Equation~\ref{eq:prob} holds because the bias of a linear form is 0 unless the linear form is identically 0, in which case its bias is 1.
\qed

\section*{Acknowledgements}

We would like to thank Alex Iosevich and Kaave Hosseini for their invaluable discussions while exploring this material, and for their advice during the writing process. This work was done while both authors were undergraduate students at the University of Rochester.

\bibliographystyle{alpha}
\bibliography{thesisbib}

\end{document}